%% file: moss18.tex
\documentclass[twocolumn]{article}
\usepackage{amsmath}
\usepackage{amssymb}

\usepackage{amsthm}
\newtheorem{lemma}{Lemma}
\newtheorem{theorem}{Theorem}
\newtheorem{corollary}{Corollary}
\usepackage{graphics}
\usepackage{url}

\newcommand{\Rule}{\mathcal{R}}
\newcommand{\Lang}{\mathcal{L}}

\newcommand{\fail}{\mathsf{fail}}
\newcommand{\n}[2][i]{\left<#2\right>_{#1}}
\newcommand{\dv}[1]{d_{c,i}\left(#1\right)}

\newcommand{\nl}{\mathord{!}}
\newcommand{\str}[1][s]{\mathbf{#1}}
\newcommand{\mt}{\boldsymbol{\mathrm{\epsilon}}}
\newcommand{\substr}[2][\str]{#1\!\left[#2\right]}
\newcommand{\chs}[1]{\mathtt{#1}}
\newcommand{\strc}[4][s]{\chs{#1}_{#2} \chs{#1}_{#3} \cdots \chs{#1}_{#4}}
\newcommand{\eos}{\#}

\newcommand{\genseq}{\alpha\beta\left[\beta_{i_1}\cdots\beta_{i_k}\right]}

\newcommand{\ie}{\textit{i.e.}}
\newcommand{\eg}{\textit{e.g.}}

\newcommand{\etal}{\textit{et~al.}}

\newcommand{\lbl}[0]{\footnotesize}

\title{Simplified Parsing Expression Derivatives}

\author{Aaron Moss\\University of Waterloo\\\texttt{a3moss@uwaterloo.ca}}

\begin{document}

\twocolumn[
\begin{@twocolumnfalse}
\maketitle

\begin{abstract}
This paper presents a new derivative parsing algorithm for parsing expression 
grammars; this new algorithm is both simpler and faster than the existing 
parsing expression derivative algorithm presented by Moss\cite{Mos17}. This 
new algorithm improves on the worst-case space and runtime bounds of 
the previous algorithm by a linear factor, as well as decreasing runtime by 
about half in practice. A proof of correctness for the new algorithm is 
included in this paper, a result not present in earlier work.
\end{abstract}

\vspace{8mm}

\end{@twocolumnfalse}
]

\section{Introduction}

A derivative parsing algorithm for parsing expression grammars (PEGs) was first 
published by Moss\cite{Mos17}; this paper \linebreak presents a simplified and 
improved algorithm, as well as a practical comparison of the two algorithms 
both to each other and to other PEG parsing methods. This new algorithm 
preserves or improves the performance bounds of the earlier algorithm, 
trimming a linear factor off the worst-case time and space bounds, 
while preserving the linear time and constant space bounds for the broad class 
of ``well-behaved'' inputs defined in \cite{Mos17}. As an additional 
contribution to the theory of parsing expression grammars, this work includes 
a formal proof of correctness for its algorithm, a result left as conjecture 
by authors of previous parsing expression derivative 
algorithms\cite{Mos17,GJWE18}. This paper also presents an extension of the 
concept of nullability from existing work on derivative 
parsing\cite{Brz64,MDS11} to PEGs, proving some useful properties of the given 
presentation while respecting Ford's\cite{For04} undecidability results.

\section{Parsing Expression Grammars} \label{defn-sec}

Parsing expression grammars are a language formalism similar in power to the 
more familiar context-free grammars (CFGs). Ford\cite{For04} has shown that 
any $\mathrm{LR}(k)$ language can be represented as a PEG, and that there are 
also some non-context free languages for which PEGs exist 
(\eg{} $a^n b^n c^n$). It is conjectured, however, that context-free languages 
exist that cannot be recognized by a PEG. Any such language would need to take 
advantage of the possible ambiguity in CFG parsing since PEGs are unambiguous 
by definition, admitting no more than one parse tree for any combination of 
grammar and input.

PEGs are a formalization of recursive-descent parsing with limited backtracking 
and infinite lookahead; (\ref{expr-eqn}) provides definitions of the 
fundamental parsing expressions. $a$ is a \emph{character literal}, matching 
and consuming a single character of input; $\varepsilon$ is the 
\emph{empty expression} which always matches without consuming any input, while 
$\varnothing$ is the \emph{failure expression}, which never matches. $A$ is a 
\emph{nonterminal}, which is replaced by its corresponding parsing expression 
$\Rule(A)$ to provide recursive structure in the formalism. The 
\emph{negative lookahead expression} $\nl\alpha$ provides much of the unique 
power of PEGs, matching only if its subexpression $\alpha$ does not match, but 
consuming no input (the \emph{positive lookahead expression} $\&\alpha$ can be 
expressed as $\nl\nl\alpha$). The \emph{sequence expression} $\alpha\beta$ 
matches $\alpha$ followed by $\beta$, while the \emph{alternation expression} 
$\alpha/\beta$ matches either $\alpha$ or $\beta$. Unlike the unordered choice 
in CFGs, if its first alternative $\alpha$ matches, an alternation expression 
never backtracks to attempt its second alternative $\beta$; this 
\emph{ordered choice} is responsible for the unambiguous nature of PEG parsing.

	\begin{equation} \label{expr-eqn}
	\begin{split}
	a(\str)            & = \begin{cases}	\str'	& \str = \chs{a}\,\str'								\\
											\fail	& \text{otherwise} 						\end{cases}	\\
	\varepsilon(\str)  & = \str																			\\
	\varnothing(\str)  & = \fail																		\\
	A(\str)            & = (\Rule(A))(\str)																\\
	\nl\alpha(\str)    & = \begin{cases}	\str	& \alpha(\str) = \fail								\\
											\fail	& \text{otherwise}						\end{cases}	\\
	\alpha\beta(\str)  & = \begin{cases}	\str''	& \alpha(\str) = \str' \wedge \beta(\str') = \str''	\\
											\fail	& \text{otherwise}						\end{cases}	\\
	\alpha/\beta(\str) & = \begin{cases}	\str'	& \alpha(\str) = \str'								\\
											\str''	& \alpha(\str) = \fail \wedge \beta(\str) = \str''	\\
											\fail	& \text{otherwise}						\end{cases}
	\end{split}
	\end{equation}

A na\"{i}ve recursive implementation of (\ref{expr-eqn}) 
may run in exponential time and linear space with respect to the input string, 
though Ford\cite{For02} has shown that a memoized \emph{packrat} implementation 
runs in linear time and space for all grammars and inputs. This linear time 
bound also suggests that some context-free languages cannot be recognized by a 
PEG, as the best known general-purpose CFG parsing algorithms have cubic 
worst-case time.

Formally, a parsing expression grammar $\mathcal{G}$ is the tuple 
$(\mathcal{N}, \mathcal{X}, \Sigma, \Rule, \sigma)$, where $\mathcal{N}$ is the 
set of nonterminals, $\mathcal{X}$ is the set of parsing expressions, $\Sigma$ 
is the input alphabet, $\Rule : \mathcal{N} \rightarrow \mathcal{X}$ maps each 
nonterminal to its corresponding parsing expression, and $\sigma \in \mathcal{X}$ 
is the start expression to parse. Each parsing expression $\varphi \in \mathcal{X}$ 
is a function $\varphi: \Sigma^* \rightarrow \Sigma^* \cup \{ \fail \}$, $\fail \notin \Sigma^*$. 
The \emph{language} $\Lang(\varphi)$ accepted by a parsing expression $\varphi$
is the set of strings matched by that parsing expression; precisely, 
$\Lang(\varphi) = \left\lbrace \str \in \Sigma^* : \exists \str' \in \Sigma^*, \varphi(\str) = \str' \right\rbrace$.
$\varphi$ is said to \emph{match} $\str$ if $\str \in \Lang(\varphi)$.
For some string $\str = \strc{i}{i+1}{n} \in \Sigma^*$, if 
$\varphi(\str) = \str' \in \Sigma^*$, $\str'$ is always a suffix 
$\strc{j}{j+1}{n}$ of $\str$; such a suffix is denoted $\substr{j}$. 
This suffix may be the entire string ($j = i$), a strict substring ($j > i$), 
or even the empty string $\mt$ (distinct from the empty expression 
$\varepsilon$). Note that for some of the new derivative expressions 
introduced in Section~\ref{deriv-sec}, $j < i$, \ie{}, $\str$ is a suffix of 
$\substr{j}$.

	\begin{equation} \label{nu-eqn}
	\begin{split}
		\nu(a)				& = \bot							\\
		\nu(\varepsilon)	& = \top							\\
		\nu(\varnothing)	& = \bot							\\
		\nu(A)				& = \nu(\Rule(A))					\\
		\nu(\nl\alpha)		& = \bot							\\
		\nu(\alpha\beta)	& = \nu(\alpha) \land \nu(\beta)	\\
		\nu(\alpha/\beta)	& = \nu(\alpha) \lor \nu(\beta)
	\end{split}
	\end{equation}

	\begin{equation} \label{lambda-eqn}
	\begin{split}
		\lambda(a)				& = \bot									\\
		\lambda(\varepsilon)	& = \top									\\
		\lambda(\varnothing)	& = \bot									\\
		\lambda(A)				& = \lambda(\Rule(A))						\\
		\lambda(\nl\alpha)		& = \top									\\
		\lambda(\alpha\beta)	& = \lambda(\alpha) \land \lambda(\beta)	\\
		\lambda(\alpha/\beta)	& = \lambda(\alpha) \lor \lambda(\beta)
	\end{split}
	\end{equation}

A parsing expression $\varphi$ is \emph{nullable} if it matches any string; 
that is, $\Lang(\varphi) = \Sigma^*$. Since 
parsing expressions match prefixes of their input, nullable expressions are 
those that match the empty string $\mt$ without constraining the rest of 
the input. Not all expressions that may match without consuming input 
(\ie{}, match $\mt$) are nullable, though; for instance, $\nl a$ \linebreak matches 
$\mt$, but not $\chs{abc}$. I define a \emph{weakly nullable} parsing 
expression $\varphi$ as one where $\mt \in \Lang(\varphi)$. 
(\ref{nu-eqn}) and (\ref{lambda-eqn}) define a nullability predicate $\nu$ and 
a weak nullability predicate $\lambda$. Both predicates can be computed by 
iteration to a fixed point.

It is worth noting that if evaluation of $(\Rule(A))(\str)$ includes a 
left-recursive call to $A(\str)$, parsing expression evaluation will never 
terminate; for this reason the behaviour of parsing expressions including 
left-recursion is undefined. Repetition of an expression can be accomplished 
right-recursively with a nonterminal of the form 
$R := \alpha\:R\;/\;\varepsilon$, generally written $\alpha*$. Fortunately, lack 
of left-recursion is easily structurally verified. The left-expansion function 
$LE$ defined in (\ref{le-eqn}) defines the set of immediate sub-expressions 
possibly left-recursively expanded by evaluation of a parsing expression. 
The left-expansion of an expression is a subset of its subexpressions $SUB$, 
defined in (\ref{sub-eqn}). The recursive left-\linebreak{}expansion and subexpression 
functions can be computed by iteration to a fixed point as 
$LE^+(\varphi) = LE(\varphi) \cup_{\gamma \in LE(\varphi)} LE^+(\gamma)$ and 
$SUB^+(\varphi) = SUB(\varphi) \cup_{\gamma \in SUB(\varphi)} SUB^+(\gamma)$. 
A parsing expression $\varphi$ is \emph{well-formed} 
if neither $\varphi$ nor any of its subexpressions left-recursively expand 
themselves; that is, 
$\forall \gamma \in \{\varphi\} \cup SUB^+(\varphi), \gamma \notin LE^+(\gamma)$. 
This definition of well-formed is equivalent to the condition for a well-formed 
grammar introduced by Ford\cite[\S~3.6]{For04}.

	\begin{equation} \label{sub-eqn}
	\begin{split}
		SUB(a)				& = \{\}					\\
		SUB(\varepsilon)	& = \{\}					\\
		SUB(\varnothing)	& = \{\}					\\
		SUB(A)				& = \lbrace\Rule(A)\rbrace	\\
		SUB(\nl\alpha)		& = \{\alpha\}				\\
		SUB(\alpha\beta)	& = \{\alpha, \beta\}		\\
		SUB(\alpha/\beta)	& = \{\alpha, \beta\}
	\end{split}
	\end{equation}

	\begin{equation} \label{le-eqn}
	\begin{split}
		LE(a)				& = \{\}					\\
		LE(\varepsilon)		& = \{\}					\\
		LE(\varnothing)		& = \{\}					\\
		LE(A)				& = \lbrace\Rule(A)\rbrace	\\
		LE(\nl\alpha)		& = \{\alpha\}				\\
		LE(\alpha\beta)		& = \begin{cases}	\{\alpha,\beta\}	& \text{if } \lambda(\alpha)	\\
												\{\alpha\} 			& \text{otherwise} \end{cases}	\\
		LE(\alpha/\beta)	& = \{\alpha, \beta\}
	\end{split}
	\end{equation}

The nullability predicate $\nu$ is a conservative under-\linebreak{}approximation of which 
parsing expressions are actually nullable; the main source of imprecision is 
lookahead expressions. For $\mu$ such that $\Lang(\mu) = \{\}$, 
$\Lang(\nl\mu) = \Sigma^*$ but $\nu(\nl\mu) = \bot$; however, Ford\cite{For04} 
showed that it is undecidable whether the language of an arbitrary parsing 
expression is empty, precluding a precise description of nullability with 
respect to lookahead expressions. As shown in Section~\ref{proof-sec}, all 
parsing expressions $\eta$ for which $\nu$ holds do match all of 
$\Sigma^*$ (Theorem~\ref{nbl-thm}), but for the converse, all that is shown is 
the weaker statment that $\lambda$ holds for any expression which matches the 
empty string (Theorem~\ref{look-thm}). The author conjectures that any significantly 
stronger statement is precluded by Ford's undecidability results.

It is sometimes useful to assume that all parsing expressions are in a 
consistent simplified form. None of the simplification rules in 
Table~\ref{simpl-table} change the result of the parsing expressions, as can be 
easily verified by consulting (\ref{expr-eqn}), while they have the useful 
properties of reducing any expression which is structurally incapable of 
matching to $\varnothing$ and of trimming unreachable or redundant 
subexpressions from sequence and alternation expressions.

\begin{table}[h]
	\caption[Simplification rules]{Simplification rules; $\nu(\eta) = \top$}
	\label{simpl-table}
	\centering
	\begin{tabular}{rr@{$\;\equiv\;$}lrr@{$\;\equiv\;$}lrr@{$\;\equiv\;$}l}
	1.	& $\alpha\varepsilon$	& $\alpha$		&
	5.	& $\alpha/\varnothing$	& $\alpha$		&
	9.  & $\nl\eta$				& $\varnothing$	\\
	2.	& $\varepsilon\beta$	& $\beta$		&
	6.	& $\varnothing/\beta$	& $\beta$		&
	10. & $\nl\varnothing$		& $\varepsilon$	\\
	3.	& $\alpha\varnothing$	& $\varnothing$	&
	7.  & $\eta/\beta$			& $\eta$		&
	11.	& $\nl\nl\nl\alpha$		& $\nl\alpha$	\\
	4.	& $\varnothing\beta$	& $\varnothing$	&
	8.	& $A$					& \multicolumn{4}{l}{\!\!\!$\varnothing$ if $A := \varnothing$}
	\end{tabular}
\end{table}

\section{Derivative Parsing}
\label{deriv-sec}

The essential idea of derivative parsing, first introduced by 
Brzozowski\cite{Brz64}, is to iteratively transform an expression into an 
expression for the ``rest'' of the input. For example, given the expression 
$\gamma = foo/bar/baz$, $d_b(\gamma) = ar/az$, the suffixes that can follow 
$\chs{b}$ in $\Lang(\gamma)$. Once repeated derivatives have been taken for 
every character in the input string, the resulting expression can be checked 
to determine whether or not it represents a match. Existing work shows how to 
compute the derivatives of regular expressions\cite{Brz64}, context-free 
grammars\cite{MDS11}, and parsing expression grammars\cite{Mos17,GJWE18}. This 
paper presents a simplified algorithm for parsing expression derivatives, as 
well as a formal proof of the correctness of this algorithm, an aspect lacking 
from the earlier presentations.

The chief difficulty in creating a derivative parsing algorithm for PEGs is 
that backtracking is a fundamental part of the semantics of PEGs. If one 
alternation branch does not match, the algorithm tries another; similarly, if a 
lookahead expression matches, the algorithm must attempt to continue parsing 
its successor from the original point in the input. However, derivative parsing 
does not backtrack -- a derivative is taken in sequence for each character, and 
all effects of parsing that character must be accounted for in the resulting 
expression. The earlier formulation by Moss\cite{Mos17} of derivative parsing 
for PEGs included a system of ``backtracking generations'' to label possible 
backtracking options for each expression, as well as a complex mapping 
algorithm to translate the backtracking generations of parsing expressions to 
the corresponding generations of their parent expressions. The key observation 
of the simplified algorithm presented here is that an index into the input 
string is sufficient to label backtracking choices consistently across all 
parsing expressions.

	\begin{equation} \label{extra-expr-eqn}
	\begin{split}
	\varepsilon_j(\substr{k})	& = \substr{j+1}													\\
	\nl_j\alpha(\substr{k})		& = \begin{cases}	\substr{j+1}	& \alpha(\substr{k}) = \fail			\\
													\fail			& \text{otherwise}	\end{cases}	\\
	\genseq(\substr{k})			& = \alpha\beta(\substr{k})
	\end{split}
	\end{equation}

In typical formulations\cite{Brz64,MDS11,Mos17}, the derivative $d_c(\varphi)$ 
is a function from an expression $\varphi \in \mathcal{X}$ and a character 
$\chs{c} \in \Sigma$ to a derivative expression $\varphi' \in \mathcal{X}$. 
Formally, $\Lang\left(d_c(\varphi)\right) = \left\{ \str \in \Sigma^* : \chs{c}\,\str \in \Lang(\varphi) \right\}$. 
This paper defines a derivative $d_{c,i}(\varphi)$, adding an index parameter 
$i \in \mathbb{N}$ for the current location in the input string. Additionally, 
certain parsing expressions are annotated with information about their input 
position. $\varepsilon$, which always matches, becomes 
$\varepsilon_j$, a match at index $j$; $\nl\alpha$, a lookahead expression 
which never consumes any characters, becomes $\nl_j\alpha$, a lookahead 
expression at index $j$. Finally, a sequence expression $\alpha\beta$ must 
track possible indices at which $\alpha$ may have stopped consuming characters 
and $\beta$ began to be parsed; to this end, $\alpha\beta$ is annotated with a 
list of \emph{lookahead followers} $\left[ \beta_{i_1} \cdots \beta_{i_k} \right]$, 
where $\beta_{i_j} \in \mathcal{X}$ is the repeated derivative of $\beta$ 
starting at each index $i_j \in \mathbb{N}$ where $\alpha$ may have stopped 
consuming characters. These annotated expressions are formally defined in 
(\ref{extra-expr-eqn}); note that these definitions match or fail under the 
same conditions as those in (\ref{expr-eqn}), but may consume (or un-consume) 
a different portion of the input, as shown in Theorem~\ref{norm-thm}. 
Accompanying extensions of $SUB$ and $LE$ are defined in (\ref{extra-sub-eqn}) 
and (\ref{extra-le-eqn}).

	\begin{equation} \label{extra-sub-eqn}
	\begin{split}
		SUB(\varepsilon_j)	& = \{\}												\\
		SUB(\nl_j\alpha)	& = \{\alpha\}											\\
		SUB(\genseq)		& = \{\alpha, \beta, \beta_{i_1}, \cdots \beta_{i_k}\}
	\end{split}
	\end{equation}

	\begin{equation} \label{extra-le-eqn}
	\begin{split}
		LE(\varepsilon_j)	& = \{\}										\\
		LE(\nl_j\alpha)		& = \{\alpha\}									\\
		LE(\genseq)			& = \{\alpha, \beta_{i_1}, \cdots \beta_{i_k}\}
	\end{split}
	\end{equation}

To annotate parsing expressions with their indices, (\ref{norm-eqn}) 
defines a \emph{normalization function} $\n{ \bullet }$ to annotate parsing 
expressions; derivative parsing of $\varphi$ starts by taking $\n[0]{\varphi}$. 
One useful effect of this normalization function is that all nonterminals in 
the left-expansion of $\varphi$ are replaced by their expansion, and thus 
$\n{\varphi}$ has no recursion in its left-expansion; in particular, structural 
induction over $(\n{\varphi})$ is bounded for well-formed $\varphi$, per 
Theorem~\ref{norm-wf-thm}.

	\begin{equation} \label{norm-eqn}
	\begin{split}
	\n{a}				& = a																			\\
	\n{\varepsilon}		& = \varepsilon_i																\\
	\n{\varnothing}		& = \varnothing																	\\
	\n{A}				& = \n{\Rule(A)}																\\
	\n{\nl\alpha}		& = \nl_i\n{\alpha}																\\
	\n{\alpha\beta}		& = \n{\alpha}\beta\left[\beta_i = \n{\beta} \mbox{ if } \lambda(\beta)\right]	\\
	\n{\alpha/\beta}	& = \n{\alpha}/\n{\beta}													
	\end{split}
	\end{equation}

The nullability functions $\nu$ and $\lambda$ must also be expanded to deal 
with the indices added by the normalization process. To this end, I define two 
functions $match$ and $back$ from $\mathcal{X}$ to $\mathcal{P}(\mathbb{N})$ 
(based on definitions in \cite{Mos17}). $match$ and $back$ may be 
thought of as extended versions of $\nu$ and $\lambda$, respectively, where 
$match$ (resp. $back$) is a non-empty set of indices if $\nu$ (resp. $\lambda$) 
is true; definitions are in (\ref{back-eqn}) and (\ref{match-eqn}) and a proof 
is included with Theorem~\ref{back-match-thm}.

	\begin{equation} \label{back-eqn}
	\begin{split}
		back(a) 				& = \{\}										\\
		back(\varepsilon_i)		& = \{i\}										\\
		back(\varnothing)		& = \{\}										\\
		back(\nl_i\alpha)		& = \{i\}										\\
		back(\genseq)			& = \cup_{j \in [i_1 \cdots i_k]} back(\beta_j)	\\
		back(\alpha/\beta)		& = back(\alpha) \cup back(\beta)
	\end{split}
	\end{equation}

	\begin{equation} \label{match-eqn}
	\begin{split}
		match(a)				& = \{\}										\\
		match(\varepsilon_i)	& = \{i\}										\\
		match(\varnothing)		& = \{\}										\\
		match(\nl_i\alpha)		& = \{\}										\\
		match(\genseq)			& = \cup_{j \in match(\alpha)} match(\beta_j)	\\
		match(\alpha/\beta)		& = match(\beta)
	\end{split}
	\end{equation}

Having defined these necessary helper functions, the derivative step function 
$d_{c,i}$ is defined in (\ref{derivative-eqn}). To test whether some 
input string $\str = \strc{1}{2}{n}$ augmented with an end-of-string 
terminal $\eos \notin \Sigma$ matches a parsing expression $\varphi$, I define 
the string derivative (\ref{str-deriv-eqn}):
	\begin{equation} \label{str-deriv-eqn}
	d_{\str,i}(\varphi) = \left(d_{s_n,i+n} \circ d_{s_{n-1},i+n-1} \circ \cdots \circ d_{s_1,i+1}\right)\left( \varphi \right).
	\end{equation}

$\varphi^{(n)} = d_{\eos,n} \circ d_{\str,0}\left( \n[0]{\varphi} \right)$ can be used to recognize 
$\Lang(\varphi)$: if $\varphi^{(n)} = \varepsilon_j$, then 
$\varphi(\str) = \substr{j}$, otherwise $\varphi(\str) = \fail$. 
This assertion is proven in 
Theorem~\ref{correct-thm}. This process may be short-circuited if some earlier 
derivative resolves to $\varphi_j$ or $\varnothing$, as Lemma~\ref{preserve-lem} 
shows these success and failure results are preserved for the rest of the 
string.

\begin{figure*}
	\centering
	\begin{equation} \label{derivative-eqn}
	\begin{split}
	\dv{a}				& = \begin{cases}	\varepsilon_i			& c = \chs{a}					\\
											\varnothing				& \mbox{otherwise}	\end{cases}	\\
	\dv{\varepsilon_j}	& = \varepsilon_j	\\
	\dv{\varnothing}	& = \varnothing		\\
	\dv{\nl_j\alpha}	& = \begin{cases}	\varnothing				& match(\dv{\alpha}) \neq \{\}				\\
											\varepsilon_j			& \dv{\alpha} = \varnothing					\\
											\nl_j\dv{\alpha}		& \mbox{otherwise}				\end{cases}	\\
	\dv{\genseq}		& = \begin{cases}	\varnothing				& \dv{\alpha} = \varnothing						\\
											\n{\beta}				& \dv{\alpha} = \varepsilon_i \land c \neq \eos	\\
											\dv{\n{\beta}}			& \dv{\alpha} = \varepsilon_i \land c = \eos	\\
											\dv{\beta_j}			& \dv{\alpha} = \varepsilon_j \land j < i		\\
											\dv{\alpha}\beta\left[\beta^\dagger_j : j \in back(\dv{\alpha})\right]
																	& \mbox{otherwise, where } \beta^\dagger_i = \n{\beta}, \beta^\dagger_{j < i} = \dv{\beta_j}	\end{cases} \\
	\dv{\alpha/\beta}	& = \begin{cases}	\dv{\beta}				& \dv{\alpha} = \varnothing										\\
											\dv{\alpha}				& \dv{\beta} = \varnothing \lor match(\dv{\alpha}) \neq \{\}	\\
											\dv{\alpha}/\dv{\beta}	& \mbox{otherwise}									\end{cases}
	\end{split}
	\end{equation}
\end{figure*}

To preserve performance, the derivative step (\ref{derivative-eqn}) should be memoized, with a 
fresh memoization table for each derivative step. With such a table, a single 
instance of $\dv{\varphi}$ can be used for all derivatives of $\varphi$, 
changing the tree of derivative expressions into a directed acyclic graph 
(DAG). The only new expressions added by computation of $\dv{\varphi}$ are 
of the form $\n{\beta}$, for $\beta$ the successor in some sequence expression 
$\alpha\beta$. Given that all such $\beta$ must be present in the original 
grammar, all of these added expressions are of constant size. With memoization, 
this bounds the increase in size of the derivative expression DAG by a constant 
factor for each derivative step. $back$ and $match$ are also memoized.

\section{Proofs} \label{proof-sec}

This paper rectifies the absence of formal rigor in the existing literature on 
parsing expression derivatives: both Moss \cite{Mos17} and the prepublication 
work of Garnock-Jones \etal{}~\cite{GJWE18} leave the correctness of their 
algorithms and predicates as conjecture with some level of experimental 
validation. This work, by contrast, includes proofs of correctness for both 
the nullability predicates and the complete algorithm.

\subsection{Nullability}

The nullability ($\nu$) and weak nullability ($\lambda$) predicates differ only 
in their treatment of lookahead expressions, and as such it is useful to 
observe that $\nu$ is a strictly stronger condition:

\begin{lemma}
	\label{nbl-look-lem}
	For any parsing expression $\varphi$, $\nu(\varphi) \Rightarrow \lambda(\varphi)$.
\end{lemma}

\begin{proof}
By cases on definitions of $\nu$ and $\lambda$.
\end{proof}

The proofs presented in this section rely heavily on the technique of 
structural induction; nonterminals present some difficulty in this approach, 
as they may recursively expand themselves, causing the induction to fail to 
terminate in a base case. The following two lemmas show that this is not an 
issue for well-formed PEGs:

\begin{lemma}
	\label{look-wf-lem}
	For any well-formed parsing expression $\varphi$, the set of parsing 
	expressions examined to compute $\lambda(\varphi)$ is precisely 
	$\{\varphi\} \cup LE^+(\varphi)$.
\end{lemma}

\begin{proof}
By structural induction on $\varphi$.
Since $\varphi$ is well-formed, there is not a nonterminal expression 
$A \in \{\varphi\} \cup SUB^+(\varphi)$ for which $A \in LE^+(A)$, therefore the 
structural induction terminates without recursively expanding any 
nonterminal. Note that this assumes the $\lambda(\alpha\beta)$ case is 
calculated in a short-circuiting manner, where $\neg\lambda(\alpha)$ implies 
$\lambda(\beta)$ is not computed.
\end{proof}

\begin{lemma}
	\label{nbl-wf-lem}
	For any well-formed parsing expression $\varphi$, the set of parsing 
	expressions examined to compute $\nu(\varphi)$ is a subset of 
	$\{\varphi\} \cup LE^+(\varphi)$.
\end{lemma}

\begin{proof}
	Structural induction as in Lemma~\ref{look-wf-lem}. 
	Note that for $\nl\alpha$ the set examined is 
	$\{\nl\alpha\} \subset \{\nl\alpha\} \cup LE^+(\nl\alpha)$ and that for 
	$\alpha\beta$, $\nu(\alpha) \Rightarrow \lambda(\alpha)$ by 
	Lemma~\ref{nbl-look-lem}.
\end{proof}

Having demonstrated the feasibility of the structural induction approach, the 
nullability results claimed in Section~\ref{defn-sec} can now be shown:

\begin{theorem}[Nullability]
	\label{nbl-thm}
	For any well-formed parsing expression $\varphi$, 
	$\nu(\varphi) \Rightarrow \Lang(\varphi) = \Sigma^*$.
\end{theorem}

\begin{proof}
	By structural induction on $\varphi$; by Lemma~\ref{nbl-wf-lem} this induction 
	does terminate. 
	The $a$, $\varnothing$, and $\nl\alpha$ cases are trivially satisfied; the 
	$\varepsilon$ and $A$ cases follow from definition.
	For $\alpha\beta$ where $\nu(\alpha) \land \nu(\beta)$, the inductive 
	hypothesis implies that for all $\str$ there exist 
	$\str[r], \str[t] \in \Sigma^*, \str = \str[rt]$, such that 
	$\alpha(\str) = \str[t]$ and $\str[t] \in \Lang(\beta)$.
	For $\alpha/\beta$ where $\nu(\alpha) \lor \nu(\beta)$, 
	if $\str \in \Sigma^* \notin \Lang(\alpha)$, by the inductive hypothesis 
	$\neg\nu(\alpha) \therefore \nu(\beta) \therefore \str \in \Lang(\beta) \therefore \str \in \Lang(\alpha/\beta)$.
\end{proof}

\begin{theorem}[Weak Nullability]
	\label{look-thm}
	For any well-formed parsing expression $\varphi$, 
	$\mt \in \Lang(\varphi) \Rightarrow \lambda(\varphi)$.
\end{theorem}

\begin{proof}
	By structural induction on contrapositive; by \linebreak 
	Lemma~\ref{look-wf-lem} this induction does terminate. 
	The $\varepsilon$ and $\nl\alpha$ cases are trivially satisfied, the $a$, 
	$\varnothing$, and $A$ cases follow from definition.
	If $\neg\lambda(\alpha\beta)$, by the inductive hypothesis 
	$\mt \in \Lang(\alpha) \Rightarrow \mt \notin \Lang(\beta) 
	\therefore \alpha(\mt) = \mt \Rightarrow \beta(\mt) = \fail
	\therefore \mt \notin \Lang(\alpha\beta)$.
	If $\neg\lambda(\alpha/\beta)$, $\neg\lambda(\alpha) \land \neg\lambda(\beta)$; 
	by the inductive hypothesis 
	$\mt \notin \Lang(\alpha) \land \mt \notin \Lang(\beta)
	\therefore \alpha(\mt) = \beta(\mt) = \alpha/\beta(\mt) = \fail$
\end{proof}

The nullability and weak nullability functions discussed in this paper are 
closely related to the $\rightharpoonup$ relation defined by 
Ford\cite[\S~3.5]{For04}. In Ford's formulation, $\varphi \rightharpoonup 0$ 
means that $\varphi$ may match while consuming no characters, 
$\varphi \rightharpoonup 1$ means that $\varphi$ may match while consuming at 
least one character, and $\varphi \rightharpoonup f$ means that $\varphi$ may 
fail to match. When recursively applied, the simplification rules in 
Table~\ref{simpl-table} ensure that any expression which is structurally 
incapable of matching\footnote{Note that this does not include expressions 
such as $\left(\nl a\right) a$, which will never match, but exceed the power 
of the rules in Table~\ref{simpl-table} to analyze.} is reduced to the 
equivalent expression $\varnothing$. Under this transformation, 
$\neg\nu(\varphi) \equiv \varphi \rightharpoonup f$, while 
$\lambda(\varphi) \equiv \varphi \rightharpoonup 0$.

\subsection{Derivatives}

As with the nullability predicates, the correctness proof for the derivative 
parsing algorithm relies heavily on structural induction, and as such must 
demonstrate the termination of that induction (Theorem~\ref{norm-wf-thm}). 
Additionally, it must be shown that the normalization step applied does not 
meaningfully change the semantics of the parsing expressions 
(Theorem~\ref{norm-thm}).

To discuss the effects of the normalization function $\n{\bullet}$, some 
terminology must be introduced. Since $\n{\bullet}$ applies to the 
left-expansion of its argument (per Lemma~\ref{norm-le-lem}), a 
\emph{normalized} parsing expression $\varphi$ is defined as a well-formed 
parsing expression with no un-normalized expressions in its left-expansion 
(\ie{}, (\ref{norm-eqn})) and where all sequence expressions 
$\genseq \in \{\varphi\} \cup LE^+(\varphi)$ respect the sequence 
normalization property (\ref{seq-norm-eqn}).
The more precise class of \emph{$k$-normalized} parsing expressions are 
normalized parsing expressions $\varphi$ with indices no greater than $k$ 
(\ie{}, (\ref{k-norm-eqn})). By contrast, an \emph{un-normalized} 
parsing expression $\varphi$ is one where (\ref{unnorm-eqn}) holds.
\begin{gather}
	\nexists \varepsilon, A, \nl\alpha, \alpha\beta \in \{\varphi\} \cup LE^+(\varphi)	\label{norm-eqn}		\\
	\lbrace i_1, i_2, \cdots, i_k \rbrace = back(\alpha)								\label{seq-norm-eqn}	\\
	\exists \varepsilon_j, \nl_j\alpha, \alpha\beta[\beta_{i_1} \cdots \beta_j] \in \{\varphi\} \cup SUB^+(\varphi) \Rightarrow j \leq k	\label{k-norm-eqn}	\\
	\nexists \varepsilon_j, \nl_j\alpha, \genseq \in \{\varphi\} \cup SUB^+(\varphi)	\label{unnorm-eqn}
\end{gather}

First I show that normalized parsing expressions have a finite expansion amenable to structural induction:

\begin{lemma}
	\label{norm-le-lem}
	$\forall \varphi \in \mathcal{X}, i \in \mathbb{N}$, the set of parsing 
	expressions expanded by $\n{\varphi}$ is precisely $LE^+(\varphi)$.
\end{lemma}

\begin{proof}
By cases on definitions of $\n{\bullet}$, $LE$.
\end{proof}


\begin{theorem}[Finite Expansion]
	\label{norm-wf-thm}
	For any well-formed, un-normalized parsing expression $\varphi$, $\n{\varphi}$ has a 
	finite expansion.
\end{theorem}

\begin{proof}
Follows directly from the definition of well-formed and Lemma~\ref{norm-le-lem}.
\end{proof}

Then I show that normalization does not change the semantics of the parsing 
expression:

\begin{theorem}[Normalization]
	\label{norm-thm}
	For any well-formed, \linebreak{}un-normalized parsing expression $\varphi$ and 
	string \linebreak{}$\str = \strc{k+1}{k+2}{k+n} \in \Sigma^*$, $\varphi(\str) = \n[k]{\varphi}(\str)$.
\end{theorem}

\begin{proof}
	By structural induction on $\varphi$.
	By Theorem~\ref{norm-wf-thm} (Finite Expansion) $\n[k]{A}$ has a finite expansion, thus the structural induction is 
	bounded. 
	$\n[k]{ \varepsilon }(\str) = \varepsilon_k(\str) = \strc{k+1}{k+2}{k+n} = \varepsilon(\str)$. 
	$\n[k]{\nl\alpha}(\str) = \nl_k\n[k]{\alpha}(\str)$; by the inductive hypothesis 
	\linebreak{} $\n[k]{\alpha}(\str) = \alpha(\str)$, thus by the definitions in (\ref{expr-eqn}) 
	and (\ref{extra-expr-eqn}) \linebreak{} $\nl_k\n[k]{\alpha}(\str) = \nl\alpha(\str)$.
	The other cases follow directly from the relevant definitions and the inductive hypothesis.
\end{proof}

In addition to showing that $\n{\varphi}$ is semantically equivalent to the 
original expression $\varphi$, the index sets $match$ and $back$ must be shown 
to represent equivalent concepts of nullability to $\nu$ and $\lambda$, 
respectively. The following lemmas present some useful properties of $match$ 
and $back$; all can be straightforwardly shown by structural induction over 
$\varphi$:

\begin{lemma}
	\label{back-match-subset-lem}
	For any normalized $\varphi \in \mathcal{X}$, $match(\varphi) \subseteq back(\varphi)$.
\end{lemma}

\begin{lemma}
	\label{back-i-subset-lem}
	For any well-formed, un-normalized $\varphi \in \mathcal{X}$, $back(\n{\varphi}) \subseteq \{i\}$.
\end{lemma}

\begin{lemma}
	\label{match-size-lem}
	For any normalized $\varphi \in \mathcal{X}$, $\left| match(\varphi) \right| \leq 1$.
\end{lemma}

Given these lemmas, the main result of equivalence to the nullability 
predicates can be shown, with the corollary that the normalization function 
$\n{\bullet}$ respects the normalization rule ($\ref{seq-norm-eqn}$):

\begin{theorem}[Nullability Equivalence]
	\label{back-match-thm}
	For a well-formed, un-normalized parsing expression $\varphi$, 
	$\nu(\varphi) \Rightarrow match(\n[k]{\varphi}) = \{k\}$ and 
	$\lambda(\varphi) \Rightarrow back(\n[k]{\varphi}) = \{k\}$.
\end{theorem}

\begin{proof}
	By structural induction on $\varphi$.
	\begin{description}
		\item|Cases $a$ and $\varnothing$| Vacuously true.
		\item|Case $\varepsilon$| Follows from definitions.
		\item|Case $A$| $\n[k]{A}$ has a finite expansion [Theorem~\ref{norm-wf-thm}], 
			thus the induction terminates. From there, case follows 
			from inductive hypothesis.
		\item|Case $\nl\alpha$| $\nu$ statement vacuously true, $\lambda$ 
			statement follows from definitions.
		\item|Case $\alpha\beta$| $\lambda(\beta)$ implies $\beta_k$ defined 
			[def'n $\n[k]{\alpha\beta}$] and $back(\beta_k) = \{k\}$ 
			[ind. hyp.]; similarly 
			$\nu(\beta) \Rightarrow \lambda(\beta)$ [Lemma~\ref{nbl-look-lem}]
			and $\nu(\beta) \Rightarrow match(\beta_k)$ [ind. hyp.]. The rest follows from definitions.
		\item|Case $\alpha/\beta$| $\nu$ statement implicitly requires 
			application of the simplification rules in Table~\ref{simpl-table}; 
			particularly that $\neg\nu(\alpha)$ [by 7.]; $\lambda$ statement 
			follows from definitions.
	\end{description}
\end{proof}

\begin{corollary}
	\label{seq-norm-cor}
	$\n{\alpha\beta}$ respects (\ref{seq-norm-eqn}).
\end{corollary}

To link up the normalization function with the idea of normalized parsing 
expressions, I prove the following link:

\begin{lemma}
	\label{k-norm-lem}
	For any well-formed, un-normalized parsing expression 
	$\varphi$, $\n[k]{\varphi}$ is $k$-normalized.
\end{lemma}

\begin{proof}
	By definition, $\nexists \varepsilon_j \mbox{ or } \nl_j\alpha \in SUB^+(\varphi)$; 
	particularly none exist such that $j > k$.
	By Lemma~\ref{norm-le-lem}, $\n[k]{\varphi}$ is applied to all of $LE^+(\varphi)$, 
	which by definition of $\n[k]{\bullet}$ replaces all of the $\varepsilon$, 
	$A$, and $\nl\alpha$ with $\varepsilon_k$, $\n[k]{\Rule(A)}$ and 
	$\nl_k\alpha$, respectively, satisfying the definitions of normalized and 
	$k$-normalized. Note $\n[k]{\Rule(A)}$ has a finite expansion [Theorem~\ref{norm-wf-thm}], 
	and (\ref{seq-norm-eqn}) is maintained [Corollary~\ref{seq-norm-cor}].
\end{proof}

With these initial results in place, I can now move on to proving the primary 
correctness result for the algorithm, as discussed in Section~\ref{deriv-sec}. 
The first step is to show that success ($\varepsilon_j$) and failure 
($\varnothing$) results persist for the rest of the string. I also show that 
the derivative step maintains the normalization property (\ref{seq-norm-eqn}), 
and that $back$ and $match$ have semantics matching their claimed meaning:

\begin{lemma}
	\label{preserve-lem}
	Given $\str = \strc{k+1}{k+2}{n}$, $d_{\str,k}(\varepsilon_j) = \varepsilon_j$
	and $d_{\str,k}(\varnothing) = \varnothing$.
\end{lemma}

\begin{proof}
	Follows directly from definitions of $\dv{\varepsilon_j}$ and \linebreak{}$\dv{\varnothing}$, applied inductively over $k$ decreasing from $n$.
\end{proof}

\begin{lemma}
	\label{deriv-back-lem}
	For any $(i-1)$-normalized expression $\varphi$, \linebreak
	$back(\dv{\varphi}) \subseteq back(\varphi) \cup \{i\}$.
\end{lemma}

\begin{proof}
	By structural induction on $\varphi$. In the $\genseq$ case
	$back(\n{\beta}) \subseteq \{i\}$ by Lemma~\ref{back-i-subset-lem}.
\end{proof}

\begin{corollary}
	$\dv{\genseq}$ maintains (\ref{seq-norm-eqn}).
\end{corollary}

\begin{lemma}
	\label{match-lem}
	For a normalized parsing expression $\varphi$ such that $match(\varphi) \neq \{\}$, 
	and a string $\str = \strc{k}{k+1}{n}$, $d_{\str,k}(\varphi) = \varepsilon_\ell$ 
	for some $\ell \leq n+1$.
\end{lemma}

\begin{proof}
	By structural induction on $\varphi$; by Theorem~\ref{norm-wf-thm} (Finite Expansion) the 
	expression is finite and thus admits structural induction.
	\begin{description}
		\item|$a$, $\varnothing$, and $\nl_j\alpha$| vacuously satisfied.
		\item|$\varepsilon_j$| Lemma~\ref{preserve-lem}.
		\item|$\genseq$ and $\alpha/\beta$| Inductive hypothesis.
	\end{description}
\end{proof}

With these lemmas established, Theorem~\ref{main-thm} shows the main result, that the derivative step possesses the expected 
semantics, while Theorem~\ref{term-thm} demonstrates how the derivative 
$d_{\eos,n}$ with respect to the end-of-string terminator fulfills the role 
typically served by a nullability combinator $\delta$ in other derivative 
parsing formulations\cite{Brz64,MDS11,GJWE18}.

\begin{theorem}[Derivative Step]
	\label{main-thm}
	For any well-formed, $k$-normalized parsing expression $\varphi$ and any string 
	$\str = \linebreak \strc{k+1}{k+2}{k+n}$, 
	$\varphi(\str) = d_{\strc{k+1}{k+2}{k+m},k}(\varphi)(\substr{k+m})$ 
	for all $m \leq n$.
\end{theorem}

\begin{proof}
	By induction on $m$. $m = 0$ is true by identity, and the inductive step is 
	shown by structural induction on $\varphi$. 
	Let $\str' = \substr{k+1}$ and for any parsing 
	expression $\gamma$ let $\gamma' = d_{\chs{s}_{k+1},k+1}(\gamma)$.
	Note that by Lemma~\ref{preserve-lem}, if 
	$\varphi' \in \{ \varepsilon_j, \varnothing \}$ then 
	$d_{\str,k}(\varphi) = \varphi'$. 
	Also note that by Lemma~\ref{match-lem} and the inductive hypothesis, 
	$match(\varphi') \neq \{\} \Rightarrow \varphi'(\str') = \substr{\ell} \Rightarrow \varphi(\str) = \substr{\ell}$.
	\begin{description}
		\item|$a$| If $\chs{s}_{k+1} = \chs{a}$, $a' = \varepsilon_{k+1}$ 
			and $a(\str) = \str'$, otherwise $a' = \varnothing$ and 
			$a(\str) = \fail$.
		\item|$\varepsilon_j$ and $\varnothing$| By Lemma~\ref{preserve-lem}.
		\item|$\nl_j\alpha$| By structural induction, 
		$\alpha(\str) = \alpha'(\str')$. 
			\begin{itemize}
				\item If $match(\alpha') \neq \{\}$, 
				$\varphi' = \varnothing \Rightarrow \varphi'(\str') = \fail$ and 
				$\alpha(\str) = \substr{\ell} \Rightarrow \varphi(\str) = \fail$.
				\item If $\alpha' = \varnothing$, 
				$\varphi' = \varepsilon_j \Rightarrow \varphi'(\str') = \substr{j}$ 
				and also $\alpha(\str) = \fail \Rightarrow \varphi(\str) = \substr{j}$.
				\item Otherwise $\varphi(\str) = \varphi'(\str')$ by 
				$\alpha(\str) = \alpha'(\str')$.
			\end{itemize}
		\item|$\alpha/\beta$| By similar argument to $\nl_j\alpha$
		\item|$\genseq$| By structural induction, $\alpha(\str) = \alpha'(\str')$.
			\begin{itemize}
				\item If $\alpha' = \varnothing$, $\varphi'(\str') = \alpha'(\str') = \fail$ and 
				also $\varphi(\str) = \alpha(\str) = \fail$. 
				\item If $\alpha' = \varepsilon_{k+1}$, 
				$\varphi'(\str') = \n[k+1]{\beta}(\str') = \beta(\str')$ [Theorem~\ref{norm-thm} (Normalization)]
				and $\alpha(\str) = \str' \therefore \varphi(\str) = \beta(\str')$.
				\item If $\alpha' = \varepsilon_j, j \leq k$, 
				$\alpha'(\str') = \alpha(\str) = \substr{j}$ and 
				$\varphi(\str) = \beta(\substr{j})$; 
				$\varphi'(\str') = \beta_j'(\str') = \n[j]{\beta}(\substr{j})$ 
				[backward application of case 5 of sequence derivative]. 
				Note that $back(\varepsilon_j) = \{j\}$ so 
				by repeated application of Lem-\linebreak{}ma~\ref{deriv-back-lem}, $\beta_j$ is defined.
				$\varphi'(\str') = \beta(\substr{j})$ [Theorem~\ref{norm-thm} (Normalization)].
				\item Otherwise $\varphi(\str) = \varphi'(\str')$ by structural induction: 
				$\alpha(\str) = \alpha'(\str')$, $\beta_j(\str) = \beta_j'(\str)$, and 
				$\beta(\str) = \beta_{k+1}^\dagger(\str')$ by Theorem~\ref{norm-thm} (Normalization) if 
				applicable.
			\end{itemize}
	\end{description}
\end{proof}

\begin{theorem}[Derivative Completion]
	\label{term-thm}
	For any string $\str = \strc{1}{2}{n} \in \Sigma^*$ and any 
	$k$-normalized expression $\varphi$, 
	$\varphi(\mt) = \substr{\ell}$ iff 
	$d_{\eos,k}( \varphi ) = \varepsilon_\ell, \ell \leq k$ and 
	$\varphi(\mt) = \fail$ iff $d_{\eos,k}( \varphi ) = \varnothing$.
\end{theorem}

\begin{proof}
	By structural induction on $\varphi$.
	\begin{description}
		\item|Case $a$| $\forall a \in \Sigma, \eos \neq a$ 
			$\therefore d_{\eos,k}(a) = \varnothing$ and $a(\mt) = \fail$.
		\item|Case $\varepsilon_j$ and $\varnothing$| By definitions.
		\item|Cases $\nl_j\alpha$ and $\alpha/\beta$| Follow directly from 
			the inductive hypothesis and definitions; note 
			$match(\varepsilon_\ell) = \{\ell\} \neq \{\} = match(\varnothing)$. \\
		\item|Case $\genseq$| 
			\begin{itemize}
				\item If $\alpha$ does not match $\mt$, then 
				$\alpha\beta(\mt) = \fail$ and 
				$d_{\eos,k}( \alpha ) = \varnothing$ by the inductive 
				hypothesis $\therefore d_{\eos,k}( \varphi ) = \varnothing$.
				\item If $\alpha(\mt) = \substr{j}$, then 
				$\alpha\beta(\mt) = \beta(\substr{j})$ and 
				$d_{\eos,k}( \alpha ) = \varepsilon_j$.
				\begin{itemize}
					\item[\textperiodcentered] If $j = k$, $d_{\eos,k}( \varphi ) = d_{\eos,k}( \n[k]{\beta} )$; 
					by Theorem~\ref{norm-thm} (Normalization), $\n[k]{\beta}(\mt) = \beta(\mt) = \beta(\substr{j})$, 
					and by Lemma~\ref{k-norm-lem}, $\n[k]{\beta}(\mt)$ is 
					$k$-normalized, therefore the inductive hypothesis applies. 
					\item[\textperiodcentered] If $j < k$, $d_{\eos,k}( \varphi ) = d_{\eos,k}( \beta_j )$; $\beta_j(\mt) = \beta(\substr{j})$ [Theorem~\ref{main-thm} (Derivative Step)].
				\end{itemize}
			\end{itemize}
	\end{description}
\end{proof}

\begin{theorem}[Derivative Correctness]
	\label{correct-thm}
	For any string $\str = \strc{1}{2}{n} \in \Sigma^*$ and 
	well-formed, un-normalized expression 
	$\varphi$, $\varphi(\str) = \substr{\ell}$ iff 
	$\varphi^{(n)} = \varepsilon_\ell$ and $\varphi(\str) = \fail$ iff 
	$\varphi^{(n)} = \varnothing$, where 
	$\varphi^{(n)} = d_{\eos,n} \circ d_{\str,0}( \n[0]{\varphi} )$.
\end{theorem}

\begin{proof}
	By Theorem~\ref{norm-thm} (Normalization) $\varphi(\str) = \n[0]{\varphi}(\str)$ and by 
	Lemma~\ref{k-norm-lem} $\n[0]{\varphi}$ is $0$-normalized, 
	$\therefore \varphi(\str) = d_{\str,0}( \n[0]{\varphi} )(\mt)$ by Theorem~\ref{main-thm} (Derivative Step);
	the result follows from Theorem~\ref{term-thm}.
\end{proof}

\section{Analysis} \label{analysis-sec}

In \cite{Mos17}, Moss demonstrated the polynomial worst-case space and time of 
his algorithm with an argument based on bounds on the depth and fanout of the 
DAG formed by his derivative expressions. These bounds, cubic space and quartic 
time, were improved to constant space and linear time for a broad class of 
``well-behaved'' inputs with constant-bounded backtracking and depth of 
recursive invocation. This paper includes a similar analysis of the algorithm 
presented herein, improving the worst-case bounds of the previous algorithm by 
a linear factor, to quadratic space and cubic time, while maintaining the 
optimal constant space and linear time \linebreak bounds for the same class of 
``well-behaved'' inputs.


For an input string of length $n$, the algorithm runs $O(n)$ derivative steps; 
the cost of each derivative step $\dv{\varphi}$ is the sum of the cost of 
(\ref{derivative-eqn}) on each expression node in $LE^+(\varphi)$. 
Since by convention the size of the grammar is a constant, all operations on 
any expression $\gamma$ from the original grammar (particularly $\n{\gamma}$) 
run in constant time and space. It can be observed from the derivative step 
and normalization equations (\ref{derivative-eqn}) and (\ref{match-eqn}) that 
once the appropriate subexpression derivatives have 
been calculated, the cost of a derivative step on a single expression node 
$\delta$ is $O(|LE(\delta)|)$. Let $b$ be the maximum $|LE(\delta)|$ over all 
$\delta \in LE^+(\varphi)$; by Lemmas~\ref{back-i-subset-lem} 
and~\ref{deriv-back-lem}, $b \in O(n)$. Assuming $\n{\bullet}$ is memoized for 
each $i$, only a constant number of expression nodes may be added to the 
expression at each derivative step, ergo $|LE^+(\varphi)| \in O(n)$. By this 
argument, the derivative parsing algorithm presented here runs in $O(n^2)$ 
worst-case space and $O(n^3)$ worst-case time, improving the previous space and 
time \linebreak bounds for derivative parsing of PEGs by a linear factor each. This linear 
improvement over the algorithm presented in \cite{Mos17} is due to the new 
algorithm only storing $O(b)$ backtracking information in sequence nodes, 
rather than $O(b^2)$ as in the previous algorithm.

In practical use, the linear time and constant space results presented 
in \cite{Mos17} for inputs with constant-bounded backtracking and grammar 
nesting (a class that includes most source code and structured data) also 
hold for this algorithm. If $b$ is bounded by a constant rather than its linear 
worst-case, the bounds discussed above are reduced to linear space and 
quadratic time. Since $b$ is a bound on the size of $LE(\varphi)$, it can be 
seen from the left-expansion equations (\ref{le-eqn}) and (\ref{extra-le-eqn}) 
that this is really a bound 
on sequence expression backtracking choices, which existing work including 
\cite{Mos17} has shown is often bounded by a constant in practical use.

Given that the bound on $b$ limits the fanout of the derivative expression DAG, 
a constant bound on the depth of that DAG implies that the overall size of the 
DAG is similarly constant-bounded. Intuitively, the bound on the depth of the 
DAG is a bound on recursive invocations of a nonterminal by itself, applying a 
sort of ``tail-call optimization'' for right-recursive invocations such as 
$R_{\alpha*} := \alpha\:R_{\alpha*}\;/\;\varepsilon$. The conjunction of both 
of these bounds defines the class of ``well-behaved'' PEG inputs introduced by 
Moss in \cite{Mos17}, and by the constant bound on derivative DAG size this 
algorithm also runs in constant space and linear time on such inputs.

\section{Experimental Results}

In addition to being easier to explain and implement than the previous 
derivative parsing algorithm, the simplified \linebreak parsing expression derivative 
presented here also has superior runtime performance.

To test this performance, the new simplified parsing expression derivative 
(SPED) algorithm was compared against the parser-combinator-based recursive 
descent (Rec.) and packrat (Pack.) parsers used in \cite{Mos17}, as 
well as the parsing expression derivative (PED) implementation from that 
paper. The same set of XML, JSON, and Java inputs and grammars used in 
\cite{Mos17} are also used here. Code and test data are available 
online\cite{Mos14}. All tests were compiled with g++ 6.2.0 and run on a machine 
with 8GB of RAM, a dual-core 2.6 GHz processor, and SSD main storage.

Figure~\ref{runtime-fig} shows the runtime of all four algorithms on all three 
data sets, plotted against the input size; Figure~\ref{mem-use-fig} shows the 
memory usage of the same runs, also plotted against the input size, but on a 
log-log scale.

\begin{figure}
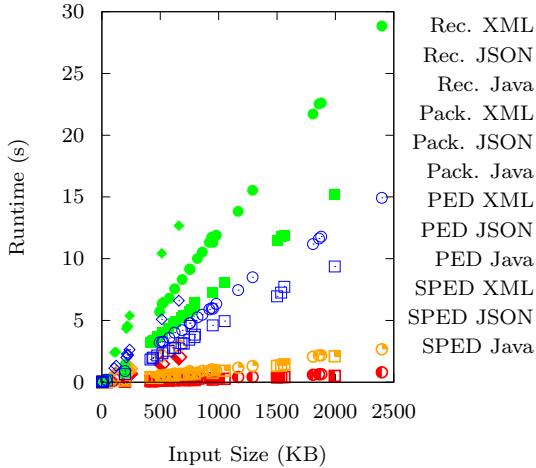

	\centering
	\include{runtime}
	\caption[Algorithm Runtime]{Algorithm runtime with respect to input size; lower is better.}
	\label{runtime-fig}
\end{figure}
	
\begin{figure}
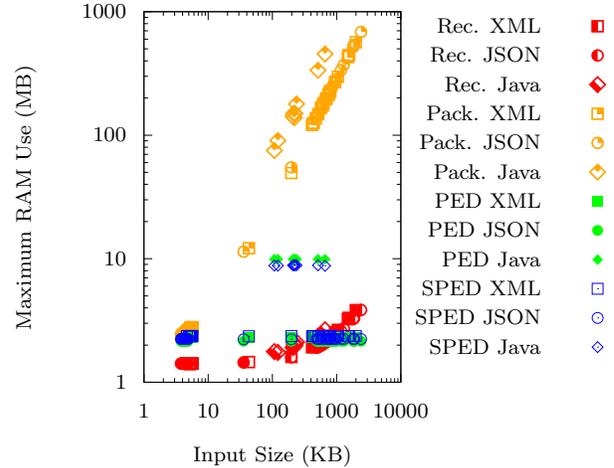

	\centering
	\include{mem-use}
	\caption[Algorithm Memory Use]{Maximum algorithm memory use with respect to input size; lower is better.}
	\label{mem-use-fig}
\end{figure}

Contrary to its poor worst-case asymptotic performance, the recursive descent 
algorithm is actually best in practice, running most quickly on all tests, and 
using the least memory on all but the largest inputs (where the derivative 
parsing algorithms' ability to not buffer input gives them an edge). Packrat 
parsing is consistently slower than recursive descent, while using two orders 
of magnitude more memory. The two derivative parsing algorithms have 
significantly \linebreak slower runtime performance, and memory usage closer to recursive 
descent than packrat.


Though on these well-behaved inputs all four algorithms run in linear time and 
space (constant space for the derivative parsing algorithms), the constant 
factor differs by both algorithm 
and grammar complexity. The XML and JSON grammars are of similar complexity, 
with 23 and 24 nonterminals, respectively, and all uses of lookahead 
expressions $\nl\alpha$ and $\&\alpha$ eliminated by judicious use of the more 
specialized negative character class, end-of-input, and until expressions 
described in \cite{Mos17}. It is consequently unsurprising that the parsers 
have similar runtime performance on those two grammars. By contrast, the Java 
grammar is significantly more complex, with 178 nonterminals and 54 lookahead 
expressions, and correspondingly poorer runtime performance.

\begin{figure}
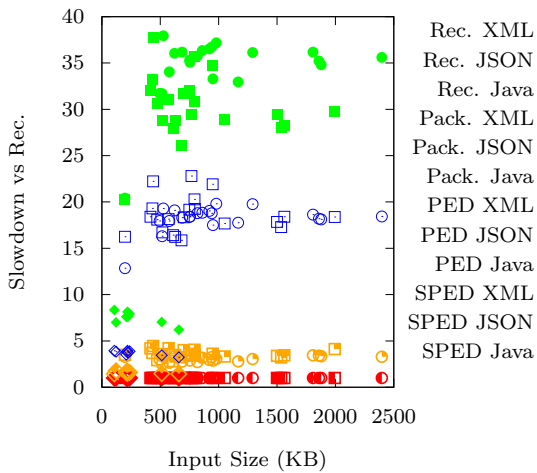

	\centering
	\include{speedup}
	\caption[Algorithm Slowdown]{Algorithm slowdown versus recursive-descent with respect to input size; lower is better.}
	\label{speedup-fig}
\end{figure}

Both the packrat algorithm and the derivative parsing algorithm presented here 
trade increased space usage for better runtime. Naturally, this trade-off works 
more in their favour for more complex grammars, particularly those with more 
lookahead expressions, as suggested by Moss\cite{Mos17}. \linebreak Grouping the broadly 
equivalent XML and JSON tests together and comparing mean speedup, recursive 
descent is 3.3x as fast as packrat and 18x as fast as SPED on XML and JSON, yet 
only 1.6x as fast as packrat and 3.7x as fast as SPED for Java. Packrat's 
runtime advantage over SPED also decreases from 5.5x to 2.3x between XML/JSON 
and Java. As can be seen from Figure~\ref{speedup-fig}, these speedup 
numbers are quite consistent across input sizes on a per-grammar basis.

Though the packrat algorithm is a modest constant factor faster than 
the derivative parsing algorithm across the test suite, it uses as much as 300x 
as much peak memory on the largest test cases, with the increases scaling 
linearly in the input size. Derivative parsing, by contrast, maintains a 
grammar-dependent constant memory usage across all the (well-behaved) inputs 
tested. This constant memory usage is within a factor of two on either side of 
the memory usage of the recursive descent implementation on all the XML and 
JSON inputs tested, and 3--5x more on the more complex Java grammar. The higher 
memory usage on Java is likely due to the lookahead expressions, which are 
handled with runtime backtracking in recursive descent, but extra 
concurrently-processed expressions in derivative parsing.

Derivative parsing in general is known to have poor runtime 
performance\cite{MDS11,AHM16}, as these results also demonstrate. However, this  
new algorithm does provide a significant improvement on the current state of 
the art for parsing expression derivatives, with a 40\% speedup on XML and 
JSON, a 50\% speedup on Java, and an up to 13\% decrease in memory usage. 
This improved performance may be beneficial for use cases that specifically 
require the derivative computation, such as the modular parsers of 
Brachth\"{a}user~\etal\cite{BRO16} or the sentence generator of 
Garnock-Jones~\etal\cite{GJWE18}.

\section{Related Work}

A number of other recognition algorithms for parsing expression grammars have 
been presented in the literature. Ford\cite{For02} introduced both the PEG 
formalism and the recursive descent and packrat algorithms. Medeiros and \linebreak
Ierusalimschy\cite{MI08} have developed a parsing machine for PEGs, similar in 
concept to a recursive descent parser, but somewhat faster in practice. 
Mizushima~\etal\cite{MMY10} have demonstrated the use of manually-inserted 
``cut operators'' to trim memory usage of packrat parsing to a constant, while 
maintaining the asymptotic worst-case bounds; Kuramitsu\cite{Kur15} and 
Redziejowski\cite{Red07} have built modified packrat parsers which use heuristic 
table-trimming mechanisms to achieve similar real-world performance without 
manual grammar modifications, but which sacrifice the polynomial worst-case 
runtime of the packrat algorithm. Henglein and Rasmussen\cite{HR17} have proved 
linear worst-case time and space bounds for their progressive tabular parsing 
algorithm, with some evidence of constant space usage in practice for a simple 
JSON grammar, but their lack of empirical comparisons to other algorithms makes 
it difficult to judge the practical utility of their approach.

Brzozowski\cite{Brz64} introduced derivative parsing for regular expressions; 
Might~\etal\cite{MDS11} extended this to context-free grammars. 
Adams~\etal\cite{AHM16} later improved the CFG derivative of Might~\etal{}, 
proving the same cubic worst-case time bound as shown here for PEG derivatives. 
Garnock-Jones \etal\cite{GJWE18} have posted a preprint of another derivative 
parsing algorithm for PEGs; their approach elegantly avoids defining new 
parsing expressions through use of a \emph{nullability combinator} to represent 
lookahead followers as later alternatives of an alternation expression. 
However, unlike this work, their paper lacks a both a proof of correctness and 
empirical performance results.

Along with defining the PEG formalism, Ford showed a number of fundamental 
theoretical results, notably that PEGs can represent any $LR(k)$ 
language\cite{For02}, and that it is undecidable in general whether the 
language of an arbitrary parsing expression is empty\cite{For04}, and by 
corollary that the equivalence of two arbitrary parsing expressions is 
undecidable. These undecidability results constrain the precision of any 
nullability function for PEGs, proving to also be a limit on the functions 
$\nu$ and $\lambda$ defined in this work. Redziejowski has also done 
significant theoretical work in 
analysis of parsing expression grammars, notably his adaptation of the $FIRST$ 
and $FOLLOW$ sets from classical parsing literature to \linebreak PEGs\cite{Red09}.

\section{Conclusion and Future Work}

This paper has proven the correctness of a new derivative parsing algorithm for 
PEGs based on the previously published algorithm in \cite{Mos17}. Its key contributions are 
simplification of the earlier algorithm through use of global numeric indices 
for backtracking choices, a proof of algorithmic correctness formerly absent 
for all published PEG derivative algorithms, and empirical comparison of this 
new algorithm to previous work. The proof of correctness includes an extension 
of the concept of \emph{nullability} from previous literature in derivative 
parsing\cite{Brz64,MDS11} to PEGs, while respecting existing undecidability 
results for PEGs\cite{For04}. The simplified algorithm also improves the 
worst-case space and time bounds of the previous algorithm by a linear factor. 

While extension of this recognition algorithm to a parsing algorithm remains 
future work, any such extension may rely on the fact that successfully 
recognized parsing expressions produce a $\varepsilon_e$ expression in this 
algorithm, where $e$ is the index at which the last character was consumed. As 
one approach, $\n[b]{\bullet}$ might annotate parsing expressions with $b$, the 
index at which they began to consume characters. By collecting subexpression 
matches and combining the two indices $b$ and $e$ on a successful match, this 
algorithm should be able to return a parse tree on match, rather than simply 
a recognition decision. The parser derivative approach of 
Might~\etal\cite{MDS11} may be useful here, with the added simplification that 
PEGs, unlike CFGs, have no more than one valid parse tree, and thus do not need 
to store multiple possible parses in a single node.

\bibliographystyle{plain}
\bibliography{peg_deriv}

\end{document}

%% file: runtime.tex
\begingroup
  \makeatletter
  \providecommand\color[2][]{%
    \GenericError{(gnuplot) \space\space\space\@spaces}{%
      Package color not loaded in conjunction with
      terminal option `colourtext'%
    }{See the gnuplot documentation for explanation.%
    }{Either use 'blacktext' in gnuplot or load the package
      color.sty in LaTeX.}%
    \renewcommand\color[2][]{}%
  }%
  \providecommand\includegraphics[2][]{%
    \GenericError{(gnuplot) \space\space\space\@spaces}{%
      Package graphicx or graphics not loaded%
    }{See the gnuplot documentation for explanation.%
    }{The gnuplot epslatex terminal needs graphicx.sty or graphics.sty.}%
    \renewcommand\includegraphics[2][]{}%
  }%
  \providecommand\rotatebox[2]{#2}%
  \@ifundefined{ifGPcolor}{%
    \newif\ifGPcolor
    \GPcolortrue
  }{}%
  \@ifundefined{ifGPblacktext}{%
    \newif\ifGPblacktext
    \GPblacktexttrue
  }{}%
  \let\gplgaddtomacro\g@addto@macro
  \gdef\gplbacktext{}%
  \gdef\gplfronttext{}%
  \makeatother
  \ifGPblacktext
    \def\colorrgb#1{}%
    \def\colorgray#1{}%
  \else
    \ifGPcolor
      \def\colorrgb#1{\color[rgb]{#1}}%
      \def\colorgray#1{\color[gray]{#1}}%
      \expandafter\def\csname LTw\endcsname{\color{white}}%
      \expandafter\def\csname LTb\endcsname{\color{black}}%
      \expandafter\def\csname LTa\endcsname{\color{black}}%
      \expandafter\def\csname LT0\endcsname{\color[rgb]{1,0,0}}%
      \expandafter\def\csname LT1\endcsname{\color[rgb]{0,1,0}}%
      \expandafter\def\csname LT2\endcsname{\color[rgb]{0,0,1}}%
      \expandafter\def\csname LT3\endcsname{\color[rgb]{1,0,1}}%
      \expandafter\def\csname LT4\endcsname{\color[rgb]{0,1,1}}%
      \expandafter\def\csname LT5\endcsname{\color[rgb]{1,1,0}}%
      \expandafter\def\csname LT6\endcsname{\color[rgb]{0,0,0}}%
      \expandafter\def\csname LT7\endcsname{\color[rgb]{1,0.3,0}}%
      \expandafter\def\csname LT8\endcsname{\color[rgb]{0.5,0.5,0.5}}%
    \else
      \def\colorrgb#1{\color{black}}%
      \def\colorgray#1{\color[gray]{#1}}%
      \expandafter\def\csname LTw\endcsname{\color{white}}%
      \expandafter\def\csname LTb\endcsname{\color{black}}%
      \expandafter\def\csname LTa\endcsname{\color{black}}%
      \expandafter\def\csname LT0\endcsname{\color{black}}%
      \expandafter\def\csname LT1\endcsname{\color{black}}%
      \expandafter\def\csname LT2\endcsname{\color{black}}%
      \expandafter\def\csname LT3\endcsname{\color{black}}%
      \expandafter\def\csname LT4\endcsname{\color{black}}%
      \expandafter\def\csname LT5\endcsname{\color{black}}%
      \expandafter\def\csname LT6\endcsname{\color{black}}%
      \expandafter\def\csname LT7\endcsname{\color{black}}%
      \expandafter\def\csname LT8\endcsname{\color{black}}%
    \fi
  \fi
    \setlength{\unitlength}{0.0500bp}%
    \ifx\gptboxheight\undefined%
      \newlength{\gptboxheight}%
      \newlength{\gptboxwidth}%
      \newsavebox{\gptboxtext}%
    \fi%
    \setlength{\fboxrule}{0.5pt}%
    \setlength{\fboxsep}{1pt}%
\begin{picture}(5102.00,3400.00)%
    \gplgaddtomacro\gplbacktext{%
      \csname LTb\endcsname%
      \put(726,340){\makebox(0,0)[r]{\strut{}\lbl 0}}%
      \put(726,806){\makebox(0,0)[r]{\strut{}\lbl 5}}%
      \put(726,1272){\makebox(0,0)[r]{\strut{}\lbl 10}}%
      \put(726,1738){\makebox(0,0)[r]{\strut{}\lbl 15}}%
      \put(726,2203){\makebox(0,0)[r]{\strut{}\lbl 20}}%
      \put(726,2669){\makebox(0,0)[r]{\strut{}\lbl 25}}%
      \put(726,3135){\makebox(0,0)[r]{\strut{}\lbl 30}}%
      \put(858,120){\makebox(0,0){\strut{}\lbl 0}}%
      \put(1298,120){\makebox(0,0){\strut{}\lbl 500}}%
      \put(1738,120){\makebox(0,0){\strut{}\lbl 1000}}%
      \put(2178,120){\makebox(0,0){\strut{}\lbl 1500}}%
      \put(2618,120){\makebox(0,0){\strut{}\lbl 2000}}%
      \put(3058,120){\makebox(0,0){\strut{}\lbl 2500}}%
    }%
    \gplgaddtomacro\gplfronttext{%
      \csname LTb\endcsname%
      \put(220,1737){\rotatebox{-270}{\makebox(0,0){\strut{}\lbl Runtime (s)}}}%
      \put(1958,-210){\makebox(0,0){\strut{}\lbl Input Size (KB)}}%
      \csname LTb\endcsname%
      \put(4114,3025){\makebox(0,0)[r]{\strut{}\lbl Rec. XML}}%
      \csname LTb\endcsname%
      \put(4114,2805){\makebox(0,0)[r]{\strut{}\lbl Rec. JSON}}%
      \csname LTb\endcsname%
      \put(4114,2585){\makebox(0,0)[r]{\strut{}\lbl Rec. Java}}%
      \csname LTb\endcsname%
      \put(4114,2365){\makebox(0,0)[r]{\strut{}\lbl Pack. XML}}%
      \csname LTb\endcsname%
      \put(4114,2145){\makebox(0,0)[r]{\strut{}\lbl Pack. JSON}}%
      \csname LTb\endcsname%
      \put(4114,1925){\makebox(0,0)[r]{\strut{}\lbl Pack. Java}}%
      \csname LTb\endcsname%
      \put(4114,1705){\makebox(0,0)[r]{\strut{}\lbl PED XML}}%
      \csname LTb\endcsname%
      \put(4114,1485){\makebox(0,0)[r]{\strut{}\lbl PED JSON}}%
      \csname LTb\endcsname%
      \put(4114,1265){\makebox(0,0)[r]{\strut{}\lbl PED Java}}%
      \csname LTb\endcsname%
      \put(4114,1045){\makebox(0,0)[r]{\strut{}\lbl SPED XML}}%
      \csname LTb\endcsname%
      \put(4114,825){\makebox(0,0)[r]{\strut{}\lbl SPED JSON}}%
      \csname LTb\endcsname%
      \put(4114,605){\makebox(0,0)[r]{\strut{}\lbl SPED Java}}%
    }%
    \gplbacktext
    \put(0,0){\includegraphics{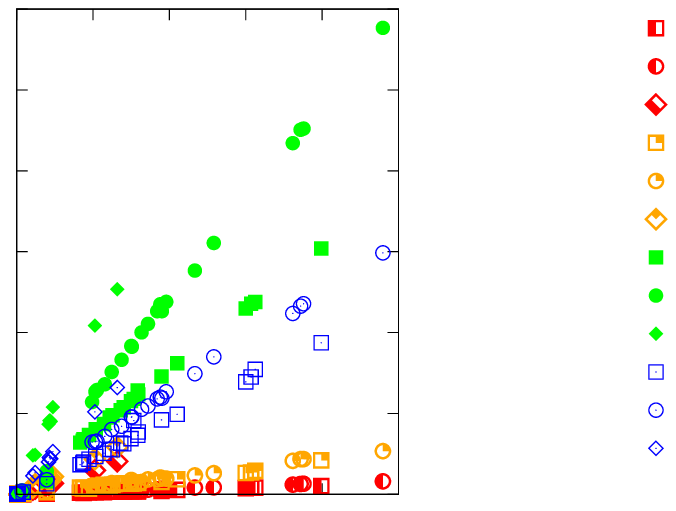}}%
    \gplfronttext
  \end{picture}%
\endgroup

%% file: mem-use.tex
\begingroup
  \makeatletter
  \providecommand\color[2][]{%
    \GenericError{(gnuplot) \space\space\space\@spaces}{%
      Package color not loaded in conjunction with
      terminal option `colourtext'%
    }{See the gnuplot documentation for explanation.%
    }{Either use 'blacktext' in gnuplot or load the package
      color.sty in LaTeX.}%
    \renewcommand\color[2][]{}%
  }%
  \providecommand\includegraphics[2][]{%
    \GenericError{(gnuplot) \space\space\space\@spaces}{%
      Package graphicx or graphics not loaded%
    }{See the gnuplot documentation for explanation.%
    }{The gnuplot epslatex terminal needs graphicx.sty or graphics.sty.}%
    \renewcommand\includegraphics[2][]{}%
  }%
  \providecommand\rotatebox[2]{#2}%
  \@ifundefined{ifGPcolor}{%
    \newif\ifGPcolor
    \GPcolortrue
  }{}%
  \@ifundefined{ifGPblacktext}{%
    \newif\ifGPblacktext
    \GPblacktexttrue
  }{}%
  \let\gplgaddtomacro\g@addto@macro
  \gdef\gplbacktext{}%
  \gdef\gplfronttext{}%
  \makeatother
  \ifGPblacktext
    \def\colorrgb#1{}%
    \def\colorgray#1{}%
  \else
    \ifGPcolor
      \def\colorrgb#1{\color[rgb]{#1}}%
      \def\colorgray#1{\color[gray]{#1}}%
      \expandafter\def\csname LTw\endcsname{\color{white}}%
      \expandafter\def\csname LTb\endcsname{\color{black}}%
      \expandafter\def\csname LTa\endcsname{\color{black}}%
      \expandafter\def\csname LT0\endcsname{\color[rgb]{1,0,0}}%
      \expandafter\def\csname LT1\endcsname{\color[rgb]{0,1,0}}%
      \expandafter\def\csname LT2\endcsname{\color[rgb]{0,0,1}}%
      \expandafter\def\csname LT3\endcsname{\color[rgb]{1,0,1}}%
      \expandafter\def\csname LT4\endcsname{\color[rgb]{0,1,1}}%
      \expandafter\def\csname LT5\endcsname{\color[rgb]{1,1,0}}%
      \expandafter\def\csname LT6\endcsname{\color[rgb]{0,0,0}}%
      \expandafter\def\csname LT7\endcsname{\color[rgb]{1,0.3,0}}%
      \expandafter\def\csname LT8\endcsname{\color[rgb]{0.5,0.5,0.5}}%
    \else
      \def\colorrgb#1{\color{black}}%
      \def\colorgray#1{\color[gray]{#1}}%
      \expandafter\def\csname LTw\endcsname{\color{white}}%
      \expandafter\def\csname LTb\endcsname{\color{black}}%
      \expandafter\def\csname LTa\endcsname{\color{black}}%
      \expandafter\def\csname LT0\endcsname{\color{black}}%
      \expandafter\def\csname LT1\endcsname{\color{black}}%
      \expandafter\def\csname LT2\endcsname{\color{black}}%
      \expandafter\def\csname LT3\endcsname{\color{black}}%
      \expandafter\def\csname LT4\endcsname{\color{black}}%
      \expandafter\def\csname LT5\endcsname{\color{black}}%
      \expandafter\def\csname LT6\endcsname{\color{black}}%
      \expandafter\def\csname LT7\endcsname{\color{black}}%
      \expandafter\def\csname LT8\endcsname{\color{black}}%
    \fi
  \fi
    \setlength{\unitlength}{0.0500bp}%
    \ifx\gptboxheight\undefined%
      \newlength{\gptboxheight}%
      \newlength{\gptboxwidth}%
      \newsavebox{\gptboxtext}%
    \fi%
    \setlength{\fboxrule}{0.5pt}%
    \setlength{\fboxsep}{1pt}%
\begin{picture}(5102.00,3400.00)%
    \gplgaddtomacro\gplbacktext{%
      \csname LTb\endcsname%
      \put(990,340){\makebox(0,0)[r]{\strut{}\lbl 1}}%
      \put(990,1272){\makebox(0,0)[r]{\strut{}\lbl 10}}%
      \put(990,2203){\makebox(0,0)[r]{\strut{}\lbl 100}}%
      \put(990,3135){\makebox(0,0)[r]{\strut{}\lbl 1000}}%
      \put(1122,120){\makebox(0,0){\strut{}\lbl 1}}%
      \put(1606,120){\makebox(0,0){\strut{}\lbl 10}}%
      \put(2090,120){\makebox(0,0){\strut{}\lbl 100}}%
      \put(2574,120){\makebox(0,0){\strut{}\lbl 1000}}%
      \put(3058,120){\makebox(0,0){\strut{}\lbl 10000}}%
    }%
    \gplgaddtomacro\gplfronttext{%
      \csname LTb\endcsname%
      \put(220,1737){\rotatebox{-270}{\makebox(0,0){\strut{}\lbl Maximum RAM Use (MB)}}}%
      \put(2090,-210){\makebox(0,0){\strut{}\lbl Input Size (KB)}}%
      \csname LTb\endcsname%
      \put(4114,3025){\makebox(0,0)[r]{\strut{}\lbl Rec. XML}}%
      \csname LTb\endcsname%
      \put(4114,2805){\makebox(0,0)[r]{\strut{}\lbl Rec. JSON}}%
      \csname LTb\endcsname%
      \put(4114,2585){\makebox(0,0)[r]{\strut{}\lbl Rec. Java}}%
      \csname LTb\endcsname%
      \put(4114,2365){\makebox(0,0)[r]{\strut{}\lbl Pack. XML}}%
      \csname LTb\endcsname%
      \put(4114,2145){\makebox(0,0)[r]{\strut{}\lbl Pack. JSON}}%
      \csname LTb\endcsname%
      \put(4114,1925){\makebox(0,0)[r]{\strut{}\lbl Pack. Java}}%
      \csname LTb\endcsname%
      \put(4114,1705){\makebox(0,0)[r]{\strut{}\lbl PED XML}}%
      \csname LTb\endcsname%
      \put(4114,1485){\makebox(0,0)[r]{\strut{}\lbl PED JSON}}%
      \csname LTb\endcsname%
      \put(4114,1265){\makebox(0,0)[r]{\strut{}\lbl PED Java}}%
      \csname LTb\endcsname%
      \put(4114,1045){\makebox(0,0)[r]{\strut{}\lbl SPED XML}}%
      \csname LTb\endcsname%
      \put(4114,825){\makebox(0,0)[r]{\strut{}\lbl SPED JSON}}%
      \csname LTb\endcsname%
      \put(4114,605){\makebox(0,0)[r]{\strut{}\lbl SPED Java}}%
    }%
    \gplbacktext
    \put(0,0){\includegraphics{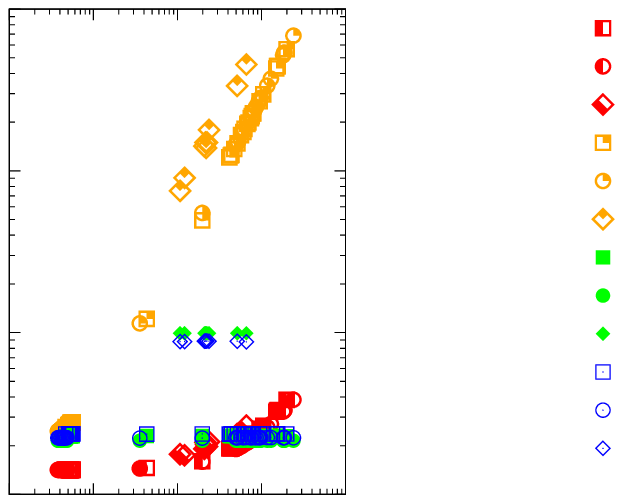}}%
    \gplfronttext
  \end{picture}%
\endgroup

%% file: speedup.tex
\begingroup
  \makeatletter
  \providecommand\color[2][]{%
    \GenericError{(gnuplot) \space\space\space\@spaces}{%
      Package color not loaded in conjunction with
      terminal option `colourtext'%
    }{See the gnuplot documentation for explanation.%
    }{Either use 'blacktext' in gnuplot or load the package
      color.sty in LaTeX.}%
    \renewcommand\color[2][]{}%
  }%
  \providecommand\includegraphics[2][]{%
    \GenericError{(gnuplot) \space\space\space\@spaces}{%
      Package graphicx or graphics not loaded%
    }{See the gnuplot documentation for explanation.%
    }{The gnuplot epslatex terminal needs graphicx.sty or graphics.sty.}%
    \renewcommand\includegraphics[2][]{}%
  }%
  \providecommand\rotatebox[2]{#2}%
  \@ifundefined{ifGPcolor}{%
    \newif\ifGPcolor
    \GPcolortrue
  }{}%
  \@ifundefined{ifGPblacktext}{%
    \newif\ifGPblacktext
    \GPblacktexttrue
  }{}%
  \let\gplgaddtomacro\g@addto@macro
  \gdef\gplbacktext{}%
  \gdef\gplfronttext{}%
  \makeatother
  \ifGPblacktext
    \def\colorrgb#1{}%
    \def\colorgray#1{}%
  \else
    \ifGPcolor
      \def\colorrgb#1{\color[rgb]{#1}}%
      \def\colorgray#1{\color[gray]{#1}}%
      \expandafter\def\csname LTw\endcsname{\color{white}}%
      \expandafter\def\csname LTb\endcsname{\color{black}}%
      \expandafter\def\csname LTa\endcsname{\color{black}}%
      \expandafter\def\csname LT0\endcsname{\color[rgb]{1,0,0}}%
      \expandafter\def\csname LT1\endcsname{\color[rgb]{0,1,0}}%
      \expandafter\def\csname LT2\endcsname{\color[rgb]{0,0,1}}%
      \expandafter\def\csname LT3\endcsname{\color[rgb]{1,0,1}}%
      \expandafter\def\csname LT4\endcsname{\color[rgb]{0,1,1}}%
      \expandafter\def\csname LT5\endcsname{\color[rgb]{1,1,0}}%
      \expandafter\def\csname LT6\endcsname{\color[rgb]{0,0,0}}%
      \expandafter\def\csname LT7\endcsname{\color[rgb]{1,0.3,0}}%
      \expandafter\def\csname LT8\endcsname{\color[rgb]{0.5,0.5,0.5}}%
    \else
      \def\colorrgb#1{\color{black}}%
      \def\colorgray#1{\color[gray]{#1}}%
      \expandafter\def\csname LTw\endcsname{\color{white}}%
      \expandafter\def\csname LTb\endcsname{\color{black}}%
      \expandafter\def\csname LTa\endcsname{\color{black}}%
      \expandafter\def\csname LT0\endcsname{\color{black}}%
      \expandafter\def\csname LT1\endcsname{\color{black}}%
      \expandafter\def\csname LT2\endcsname{\color{black}}%
      \expandafter\def\csname LT3\endcsname{\color{black}}%
      \expandafter\def\csname LT4\endcsname{\color{black}}%
      \expandafter\def\csname LT5\endcsname{\color{black}}%
      \expandafter\def\csname LT6\endcsname{\color{black}}%
      \expandafter\def\csname LT7\endcsname{\color{black}}%
      \expandafter\def\csname LT8\endcsname{\color{black}}%
    \fi
  \fi
    \setlength{\unitlength}{0.0500bp}%
    \ifx\gptboxheight\undefined%
      \newlength{\gptboxheight}%
      \newlength{\gptboxwidth}%
      \newsavebox{\gptboxtext}%
    \fi%
    \setlength{\fboxrule}{0.5pt}%
    \setlength{\fboxsep}{1pt}%
\begin{picture}(5102.00,3400.00)%
    \gplgaddtomacro\gplbacktext{%
      \csname LTb\endcsname%
      \put(726,340){\makebox(0,0)[r]{\strut{}\lbl 0}}%
      \put(726,689){\makebox(0,0)[r]{\strut{}\lbl 5}}%
      \put(726,1039){\makebox(0,0)[r]{\strut{}\lbl 10}}%
      \put(726,1388){\makebox(0,0)[r]{\strut{}\lbl 15}}%
      \put(726,1738){\makebox(0,0)[r]{\strut{}\lbl 20}}%
      \put(726,2087){\makebox(0,0)[r]{\strut{}\lbl 25}}%
      \put(726,2436){\makebox(0,0)[r]{\strut{}\lbl 30}}%
      \put(726,2786){\makebox(0,0)[r]{\strut{}\lbl 35}}%
      \put(726,3135){\makebox(0,0)[r]{\strut{}\lbl 40}}%
      \put(858,120){\makebox(0,0){\strut{}\lbl 0}}%
      \put(1298,120){\makebox(0,0){\strut{}\lbl 500}}%
      \put(1738,120){\makebox(0,0){\strut{}\lbl 1000}}%
      \put(2178,120){\makebox(0,0){\strut{}\lbl 1500}}%
      \put(2618,120){\makebox(0,0){\strut{}\lbl 2000}}%
      \put(3058,120){\makebox(0,0){\strut{}\lbl 2500}}%
    }%
    \gplgaddtomacro\gplfronttext{%
      \csname LTb\endcsname%
      \put(220,1737){\rotatebox{-270}{\makebox(0,0){\strut{}\lbl Slowdown vs Rec.}}}%
      \put(1958,-210){\makebox(0,0){\strut{}\lbl Input Size (KB)}}%
      \csname LTb\endcsname%
      \put(4114,3025){\makebox(0,0)[r]{\strut{}\lbl Rec. XML}}%
      \csname LTb\endcsname%
      \put(4114,2805){\makebox(0,0)[r]{\strut{}\lbl Rec. JSON}}%
      \csname LTb\endcsname%
      \put(4114,2585){\makebox(0,0)[r]{\strut{}\lbl Rec. Java}}%
      \csname LTb\endcsname%
      \put(4114,2365){\makebox(0,0)[r]{\strut{}\lbl Pack. XML}}%
      \csname LTb\endcsname%
      \put(4114,2145){\makebox(0,0)[r]{\strut{}\lbl Pack. JSON}}%
      \csname LTb\endcsname%
      \put(4114,1925){\makebox(0,0)[r]{\strut{}\lbl Pack. Java}}%
      \csname LTb\endcsname%
      \put(4114,1705){\makebox(0,0)[r]{\strut{}\lbl PED XML}}%
      \csname LTb\endcsname%
      \put(4114,1485){\makebox(0,0)[r]{\strut{}\lbl PED JSON}}%
      \csname LTb\endcsname%
      \put(4114,1265){\makebox(0,0)[r]{\strut{}\lbl PED Java}}%
      \csname LTb\endcsname%
      \put(4114,1045){\makebox(0,0)[r]{\strut{}\lbl SPED XML}}%
      \csname LTb\endcsname%
      \put(4114,825){\makebox(0,0)[r]{\strut{}\lbl SPED JSON}}%
      \csname LTb\endcsname%
      \put(4114,605){\makebox(0,0)[r]{\strut{}\lbl SPED Java}}%
    }%
    \gplbacktext
    \put(0,0){\includegraphics{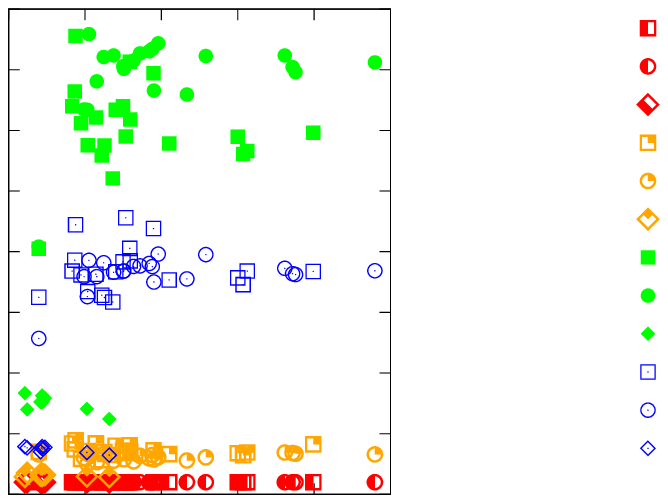}}%
    \gplfronttext
  \end{picture}%
\endgroup